\documentclass[a4paper,11pt]{amsart}
\usepackage{fancyhdr}
\usepackage{amsmath}
\usepackage{dsfont}
\usepackage{hyperref}
\usepackage[mathscr]{eucal}
\usepackage[cp1251]{inputenc}
\usepackage[english]{babel}
\usepackage{enumerate,float,indentfirst}
\usepackage{graphicx}
\usepackage{xcolor}
\usepackage{latexsym,a4,mathrsfs,amsthm,amsmath,amssymb,url}
\usepackage{amsfonts}
\usepackage{amssymb}

\usepackage{tikz}
\usepackage{pgfplots}
\usetikzlibrary{plotmarks}

\numberwithin{equation}{section}
\newtheorem{lemma}{Lemma}

\begin{document}

\vspace{1in}

\title[On spherical harmonics possessing octahedral symmetry]{\bf On spherical harmonics possessing octahedral symmetry}


\author[Yu. Nesterenko]{Yu. Nesterenko }
\address{ Mechanical Analysis Division \\
 Mentor, a Siemens Business }
\email{Yuri\_Nesterenko@mentor.com}


\begin{abstract}
In this paper, we present the implicit equations for one special class of real-valued spherical harmonics with octahedral symmetry.
Based on this representation, we construct the rotationally invariant measure of deviation from the specified symmetry.
The spherical harmonics we consider have some applications in the area of directional fields design due to their
ability to represent mutually orthogonal axes in 3D space, not relative to their order and orientation.
\end{abstract}

\maketitle

\thispagestyle{empty}

\section{Introduction}

We consider real-valued spherical harmonics of degree 4 on the unit sphere.
These functions form 9D space with standard orthonormal basis $Y_{4,-4}, \ldots, Y_{4,4}$ (see \cite{Gorller1996}).
\begin{center}
\begin{figure}[h]
\includegraphics[width=13.0cm]{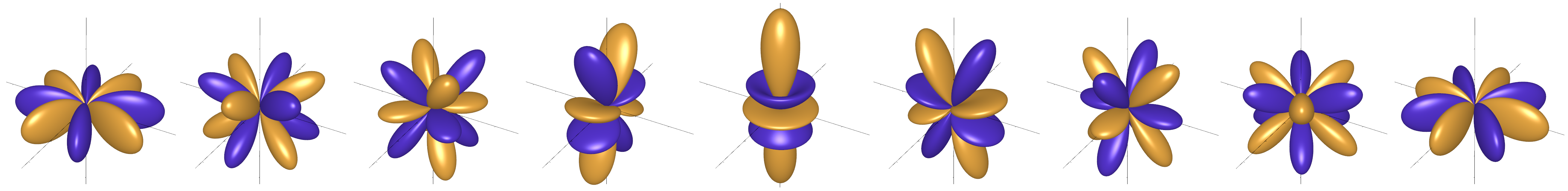}
\caption{ Spherical plots of basis functions $Y_{4,-4}, \ldots, Y_{4,4}$. }
\label{fig:basis}
\end{figure}
\end{center}

The object of study in this work is the 3D manifold of all normalized spherical harmonics possessing octahedral symmetry.
Up to multiplication by $-1$ all harmonics of this kind may be obtained from the reference harmonic $\tilde{h}$ with coordinates
\begin{equation*}
\tilde{a} = (0,0,0,0,\sqrt{\frac{7}{12}},0,0,0,\sqrt{\frac{5}{12}})^T \in \mathds{R}^9,
\end{equation*}
by rotations
\begin{equation*}
a = R_x(\alpha) \times R_y(\beta) \times R_z(\gamma) \times \tilde{a},
\end{equation*}
where $\alpha, \beta, \gamma$ are Euler angles (as well as in cases of different degrees the space we consider has an important property:
it is closed under 3D rotations, i.e. applying a rotation to a harmonic of degree 4 produces another harmonic of the same degree).
Appendix A.1 describes the construction of the rotation matrices $R_x$, $R_y$ and $R_z$.
 \begin{center}
\begin{figure}
\includegraphics[width=8.0cm]{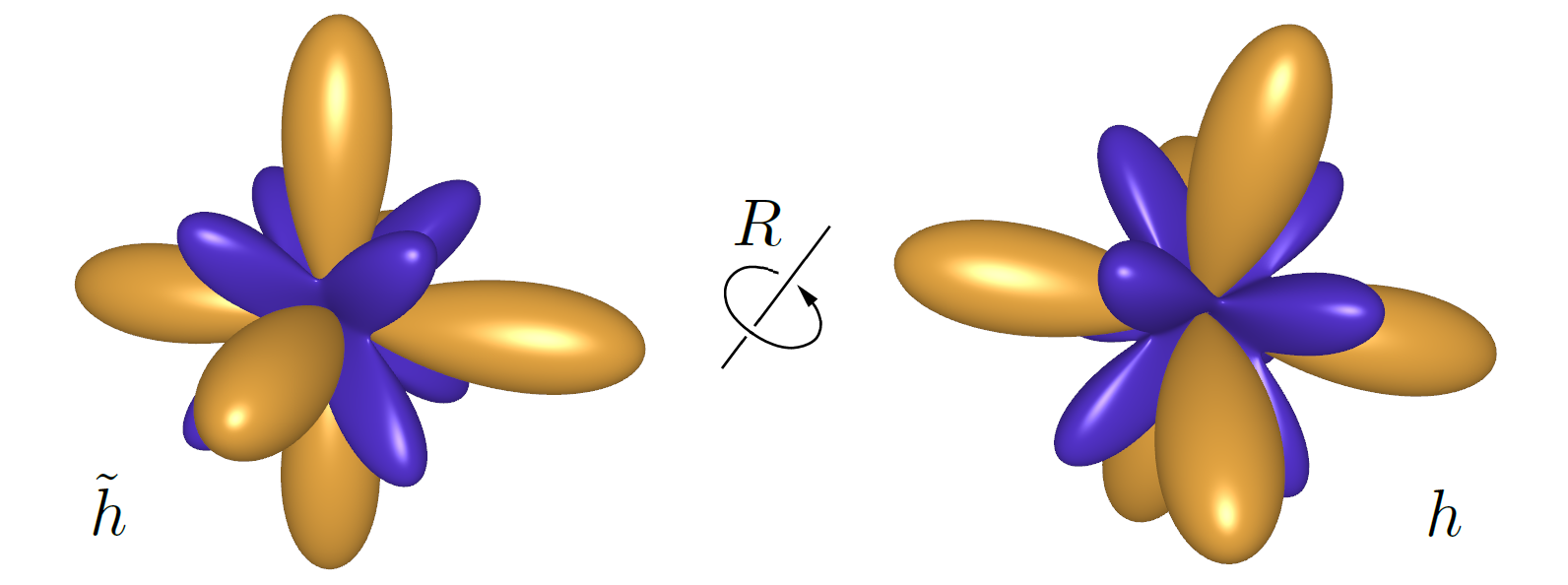}
\caption{ The reference harmonic and its rotation. }
\label{fig:def}
\end{figure}
\end{center}
From geometrical point of view, $a(\alpha, \beta, \gamma)$ is constrained to be on
the manifold of dimension 3 embedded in the $\mathds{R}^9$.
Appendix A.2 contains description of topology of this manifold.

The need of effective manipulation with the spherical harmonics of the specified class was recently realized in context of the problem of directional fields smoothing (see \cite{Beaufort2019,Huang2011,Palmer2020,Ray2015,Solomon2017})
due to their ability to represent mutually orthogonal axes in 3D space, not relative to their order and orientation.

\section{Results}
The next lemma claims that harmonics manifold which we study is simply intersection of quadrics (hypersurfaces of the second order).
\begin{lemma}\label{l1}
The manifold of all normalized spherical harmonics possessing octahedral symmetry is given by the system of equations
\begin{equation}\label{inp}
\left\{ \begin{split}
&a^T \, a = 1, \\
&a^T \, S_k \, a = 0, \quad k = 1,\ldots,5,
\end{split}\right.
\end{equation}
where $S_1, \ldots, S_5$ are the symmetric matrices defined as follows.
\begin{Small}
\begin{equation}\label{s1}
S_1 = \sqrt{2}
\begin{bmatrix}
 28           & 0             & 0             & 0             & 0             & 0             & 0             & 0             & 0              \\
0             &  7            & 0             & 0             & 0             & 0             & 0             & 0             & 0              \\
0             & 0             & -8            & 0             & 0             & 0             & 0             & 0             & 0              \\
0             & 0             & 0             & -17           & 0             & 0             & 0             & 0             & 0              \\
0             & 0             & 0             & 0             & -20           & 0             & 0             & 0             & 0              \\
0             & 0             & 0             & 0             & 0             & -17           & 0             & 0             & 0              \\
0             & 0             & 0             & 0             & 0             & 0             & -8            & 0             & 0              \\
0             & 0             & 0             & 0             & 0             & 0             & 0             &  7            & 0              \\
0             & 0             & 0             & 0             & 0             & 0             & 0             & 0             &  28
\end{bmatrix}
\end{equation}

\begin{equation}\label{s2}
S_2 = \sqrt{3}
\begin{bmatrix}
0             &  14           & 0             & 0             & 0             & 0             & 0             & 0             & 0              \\
 14           & 0             &  5\sqrt{7}    & 0             & 0             & 0             & 0             & 0             & 0              \\
0             &  5\sqrt{7}    & 0             &  9            & 0             & 0             & 0             & 0             & 0              \\
0             & 0             &  9            & 0             & 0             & 0             & 0             & 0             & 0              \\
0             & 0             & 0             & 0             & 0             &  2\sqrt{5}    & 0             & 0             & 0              \\
0             & 0             & 0             & 0             &  2\sqrt{5}    & 0             &  9            & 0             & 0              \\
0             & 0             & 0             & 0             & 0             &  9            & 0             &  5\sqrt{7}    & 0              \\
0             & 0             & 0             & 0             & 0             & 0             &  5\sqrt{7}    & 0             &  14            \\
0             & 0             & 0             & 0             & 0             & 0             & 0             &  14           & 0
\end{bmatrix}
\end{equation}

\begin{equation}\label{s3}
S_3 = \sqrt{3}
\begin{bmatrix}
0             & 0             & 0             & 0             & 0             & 0             & 0             &  14           & 0              \\
0             & 0             & 0             & 0             & 0             & 0             &  5\sqrt{7}    & 0             & -14            \\
0             & 0             & 0             & 0             & 0             &  9            & 0             & -5\sqrt{7}    & 0              \\
0             & 0             & 0             & 0             &  2\sqrt{5}    & 0             & -9            & 0             & 0              \\
0             & 0             & 0             &  2\sqrt{5}    & 0             & 0             & 0             & 0             & 0              \\
0             & 0             &  9            & 0             & 0             & 0             & 0             & 0             & 0              \\
0             &  5\sqrt{7}    & 0             & -9            & 0             & 0             & 0             & 0             & 0              \\
 14           & 0             & -5\sqrt{7}    & 0             & 0             & 0             & 0             & 0             & 0              \\
0             & -14           & 0             & 0             & 0             & 0             & 0             & 0             & 0
\end{bmatrix}
\end{equation}

\begin{equation}\label{s4}
S_4 = \sqrt{6}
\begin{bmatrix}
0             & 0             &  2\sqrt{7}    & 0             & 0             & 0             & 0             & 0             & 0              \\
0             & 0             & 0             &  3\sqrt{7}    & 0             & 0             & 0             & 0             & 0              \\
 2\sqrt{7}    & 0             & 0             & 0             & 0             & 0             & 0             & 0             & 0              \\
0             &  3\sqrt{7}    & 0             &  10           & 0             & 0             & 0             & 0             & 0              \\
0             & 0             & 0             & 0             & 0             & 0             &  6\sqrt{5}    & 0             & 0              \\
0             & 0             & 0             & 0             & 0             & -10           & 0             &  3\sqrt{7}    & 0              \\
0             & 0             & 0             & 0             &  6\sqrt{5}     & 0             & 0             & 0             &  2\sqrt{7}    \\
0             & 0             & 0             & 0             & 0             &  3\sqrt{7}    & 0             & 0             & 0              \\
0             & 0             & 0             & 0             & 0             & 0             &  2\sqrt{7}    & 0             & 0
\end{bmatrix}
\end{equation}

\begin{equation}\label{s5}
S_5 = \sqrt{6}
\begin{bmatrix}
0             & 0             & 0             & 0             & 0             & 0             & 2\sqrt{7}     & 0             & 0              \\
0             & 0             & 0             & 0             & 0             & 3\sqrt{7}     & 0             & 0             & 0              \\
0             & 0             & 0             & 0             & 6\sqrt{5}     & 0             & 0             & 0             & -2\sqrt{7}     \\
0             & 0             & 0             & 0             & 0             & -10           & 0             & -3\sqrt{7}    & 0              \\
0             & 0             & 6\sqrt{5}     & 0             & 0             & 0             & 0             & 0             & 0              \\
0             & 3\sqrt{7}     & 0             & -10           & 0             & 0             & 0             & 0             & 0              \\
2\sqrt{7}     & 0             & 0             & 0             & 0             & 0             & 0             & 0             & 0              \\
0             & 0             & 0             & -3\sqrt{7}    & 0             & 0             & 0             & 0             & 0              \\
0             & 0             & -2\sqrt{7}    & 0             & 0             & 0             & 0             & 0             & 0
\end{bmatrix}
\end{equation}
\end{Small}

\end{lemma}

\emph{Idea of proof.}
The given implicit equations were obtained by the standard technique based on rational parametrization of the unit circle and Gr\"obner basis construction.
It can be verified by direct calculations.

\quad

Matrices (\ref{s1}) - (\ref{s5}) were chosen among their different possible linear combinations based on the following additional consideration.
\begin{lemma}\label{l2}
Real-valued function
\begin{equation}\label{dev}
d(a) = \sum_{k = 1}^{5} (a^T \, S_k \, a)^2,
\end{equation}
where $a \in \mathds{R}^9$, defines rotationally invariant measure of harmonic's deviation from octahedral symmetry.
\end{lemma}

\emph{Idea of proof.}
The result was obtained by averaging of some trial non-invariant deviation measure over $SO_3$ group.
As in the previous lemma, rotational invariance of the specified function can be verified by direct calculations.

\section{Numerical example}
In this section we show how deviation measure (\ref{dev}) works.
We use the combination of the scale controlling term and the symmetry deviation term with positive weights $w_1$ and $w_2$
\begin{equation}\label{energy}
p(a; w_1, w_2) = w_1 (a^T \, a - 1)^2 + w_2 \sum_{k = 1}^{5} (a^T \, S_k \, a)^2, \quad w_1, \, w_2 > 0,
\end{equation}
as the penalty function for simple gradient descent method.
Figure \ref{fig:harms} shows the convergence process of the sample initial harmonic to the symmetrical one.
\begin{center}
\begin{figure}
\includegraphics[width=12cm]{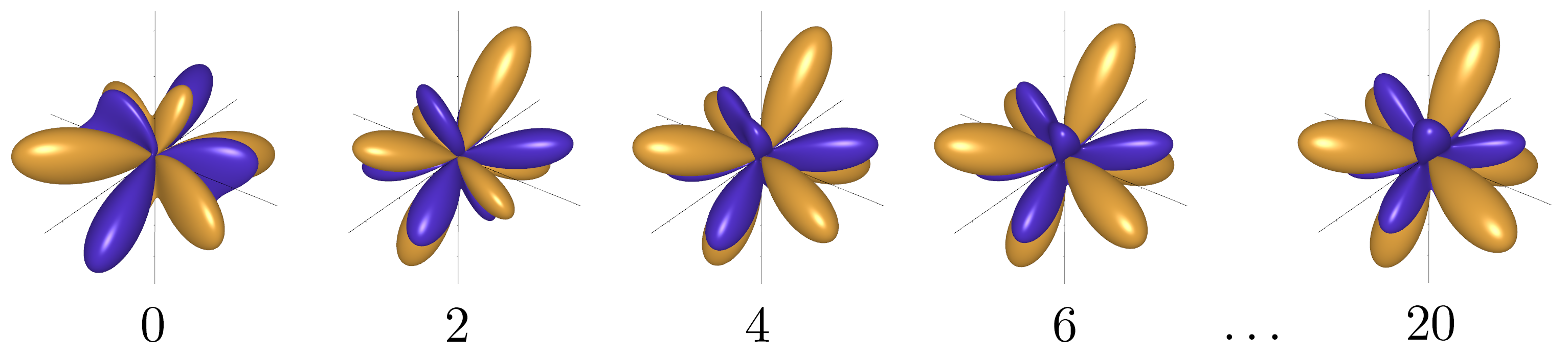}
\caption{ Iterative symmetrization }
\label{fig:harms}
\end{figure}
\end{center}

The plots below describe the distance to the nearest symmetrical harmonic and square root of the penalty value.
One can see distance-like behavior of square root of $p(a; w_1, w_2)$.

\quad

\begin{center}
\begin{tikzpicture}
    \begin{axis}[
        height=0.45\textwidth,
        width=0.7\textwidth,
        xlabel=Iterations,
        xtick={0,2,...,20},
        ylabel=Distance measures,
        ymode=log,
       log basis y={10}
]
        \addplot[mark=*,mark options={fill=white},black] coordinates {(0	,0.293288)
(1	,0.2555)
(2	,0.2135)
(3	,0.198)
(4	,0.177)
(5	,0.1545)
(6	,0.127)
(7	,0.0995)
(8	,0.072)
(9	,0.04945)
(10 ,0.03235)
(11 ,0.02065)
(12 ,0.013)
(13 ,0.00815)
(14 ,0.0051)
(15 ,0.003195)
(16 ,0.002005)
(17 ,0.00126)
(18 ,0.000795)
(19 ,0.0005)
(20 ,0.0003175)};

\addplot[mark=triangle*,mark options={fill=white},black] coordinates {(0, 0.775715798 )
 (1, 0.611555394 )
 (2, 0.386005181 )
 (3, 0.302489669 )
 (4, 0.282842712 )
 (5, 0.258650343 )
 (6, 0.225388553 )
 (7, 0.183575598 )
 (8, 0.13820275  )
 (9, 0.096747093 )
 (10, 0.064031242)
 (11, 0.04110961 )
 (12, 0.025980762)
 (13, 0.016309506)
 (14, 0.010198039)
 (15, 0.006403124)
 (16, 0.004024922)
 (17, 0.002527845)
 (18, 0.001593738)
 (19, 0.001004988)
(20, 0.000637966)};
    \end{axis}

   	\begin{scope}[shift={(4.5,3.2)}]
	\draw (0,0) --
		plot[mark=*, mark options={fill=white}] (0.25,0) -- (0.5,0)
		node[right]{sqrt penalty};
	\draw[yshift=\baselineskip] (0,0) --
		plot[mark=triangle*, mark options={fill=white}] (0.25,0) -- (0.5,0)
		node[right]{distance};
	\end{scope}
\end{tikzpicture}
\end{center}
Penalty (\ref{energy}) together with squared gradient term gives a kind of Ginzburg-Landau energy for smoothing of spherical harmonics fields in 3D (see \cite{Macq2020,Viertel2019}).

\section{Conclusion}
The implicit equations for the manifold of normalized spherical harmonics possessing octahedral symmetry are found.
Based on this representation the rotationally invariant measure of harmonic's deviation from considered symmetry is constructed.
The numerical example illustrating the behavior of the constructed deviation measure is given.
The obtained results have some applications in the area of directional fields design.
\newpage

\section{Appendix A.1}
The rotational matrices $R_x$, $R_y$ and $R_z$ for spherical harmonics of degree 4 are defined as follows.

\begin{Small}
\begin{equation}
R_z(\gamma) =
\begin{bmatrix}
\cos 4\gamma  & 0             & 0             & 0             & 0             & 0             & 0             & 0             &  \sin 4\gamma  \\
0             & \cos 3\gamma  & 0             & 0             & 0             & 0             & 0             &  \sin 3\gamma & 0              \\
0             & 0             & \cos 2\gamma  & 0             & 0             & 0             &  \sin 2\gamma & 0             & 0              \\
0             & 0             & 0             & \cos  \gamma  & 0             &  \sin  \gamma & 0             & 0             & 0              \\
0             & 0             & 0             & 0             & 1             & 0             & 0             & 0             & 0              \\
0             & 0             & 0             & -\sin  \gamma & 0             & \cos  \gamma  & 0             & 0             & 0              \\
0             & 0             & -\sin 2\gamma & 0             & 0             & 0             & \cos 2\gamma  & 0             & 0              \\
0             & -\sin 3\gamma & 0             & 0             & 0             & 0             & 0             & \cos 3\gamma  & 0              \\
-\sin 4\gamma & 0             & 0             & 0             & 0             & 0             & 0             & 0             & \cos 4\gamma
\end{bmatrix}
\end{equation}

\begin{equation}
R_x(\frac{\pi}{2}) = \frac{1}{8}
\begin{bmatrix}
0             & 0             & 0             & 0             & 0             &  2\sqrt{14}   & 0             & -2\sqrt{2}    & 0              \\
0             & -6            & 0             &  2\sqrt{7}    & 0             & 0             & 0             & 0             & 0              \\
0             & 0             & 0             & 0             & 0             &  2\sqrt{2}    & 0             &  2\sqrt{14}   & 0              \\
0             &  2\sqrt{7}    & 0             &  6            & 0             & 0             & 0             & 0             & 0              \\
0             & 0             & 0             & 0             &  3            & 0             &  2\sqrt{5}    & 0             &  \sqrt{35}     \\
-2\sqrt{14}   & 0             & -2\sqrt{2}    & 0             & 0             & 0             & 0             & 0             & 0              \\
0             & 0             & 0             & 0             &  2\sqrt{5}    & 0             &  4            & 0             & -2\sqrt{7}     \\
 2\sqrt{2}    & 0             & -2\sqrt{14}   & 0             & 0             & 0             & 0             & 0             & 0              \\
0             & 0             & 0             & 0             &  \sqrt{35}    & 0             & -2\sqrt{7}    & 0             &  1
\end{bmatrix}
\end{equation}
\end{Small}

\begin{equation}
R_y(\beta) = R_x(\frac{\pi}{2}) \times R_z(\beta) \times R_x(\frac{\pi}{2})^T
\end{equation}

\begin{equation}
R_x(\alpha) = R_y(\frac{\pi}{2})^T \times R_z(\alpha) \times R_y(\frac{\pi}{2})
\end{equation}

See \cite{Blanco1997,Choi1999,Collado1989,Ivanic1996} for more details.

\newpage

\section{Appendix A.2}
Topologically, the manifold of all rotations of the reference harmonic is the quotient space $SO_3 \, / \, S_4$, where $S_4$ denotes the group of order 24 of all octahedra symmetries.

It may be clearly described using the Rodriguez representation of $3D$ rotations (see \cite{Becker1989}).
The fundamental zone for octahedral symmetry in this representation has the form of a truncated cube with 6 regular octagonal faces and 8 regular triangular faces, as shown in the picture below.

\begin{center}
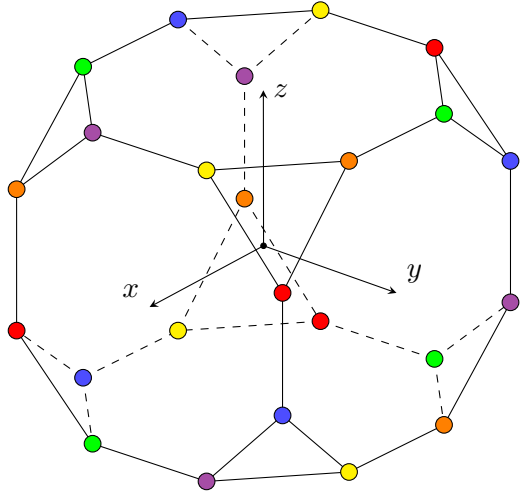
\begin{figure}[h]
\begin{tikzpicture}
\def \d {0.025cm}
\def \r {0.11cm}

\draw[xshift=160 * \d, yshift=-160 * \d]

    (-130 * \d, 30 * \d) -- (-130 * \d, -45 * \d) -- (-90 * \d, -105 * \d) -- (-30 * \d, -125 * \d) -- (45 * \d, -120 * \d) -- (95 * \d, -95 * \d) --
    (130 * \d, -30 * \d) -- (130 * \d, 45 * \d) -- (90 * \d, 105 * \d) -- (30 * \d, 125 * \d) -- (-45 * \d, 120 * \d) -- (-95 * \d, 95 * \d) -- (-130 * \d, 30 * \d)

    (-90 * \d, 60 * \d) -- (-130 * \d, 30 * \d)
    (-90 * \d, 60 * \d) -- (-95 * \d, 95 * \d)
    (-90 * \d, 60 * \d) -- (-30 * \d, 40 * \d)

    (10 * \d, -90 * \d) -- (-30 * \d, -125 * \d)
    (10 * \d, -90 * \d) -- (45 * \d, -120 * \d)
    (10 * \d, -90 * \d) -- (10 * \d, -25 * \d)

    (95 * \d, 70 * \d) -- (130 * \d, 45 * \d)
    (95 * \d, 70 * \d) -- (90 * \d, 105 * \d)
    (95 * \d, 70 * \d) -- (45 * \d, 45 * \d)

    (10 * \d, -25 * \d) -- (-30 * \d, 40 * \d)
    (45 * \d, 45 * \d) -- (10 * \d, -25 * \d)
    (-30 * \d, 40 * \d) -- (45 * \d, 45 * \d)
;

\draw[xshift=160 * \d, yshift=-160 * \d]
[dashed]

    (90 * \d, -60 * \d) -- (130 * \d, -30 * \d)
    (90 * \d, -60 * \d) -- (95 * \d, -95 * \d)
    (90 * \d, -60 * \d) -- (30 * \d, -40 * \d)

    (-10 * \d, 90 * \d) -- (30 * \d, 125 * \d)
    (-10 * \d, 90 * \d) -- (-45 * \d, 120 * \d)
    (-10 * \d, 90 * \d) -- (-10 * \d, 25 * \d)

    (-95 * \d, -70 * \d) -- (-130 * \d, -45 * \d)
    (-95 * \d, -70 * \d) -- (-90 * \d, -105 * \d)
    (-95 * \d, -70 * \d) -- (-45 * \d, -45 * \d)

    (-10 * \d, 25 * \d) -- (30 * \d, -40 * \d)
    (-45 * \d, -45 * \d) -- (-10 * \d, 25 * \d)
    (30 * \d, -40 * \d) -- (-45 * \d, -45 * \d)
;

\begin{scope}[shift={(160 * \d,-160 * \d)}]
\draw[fill=orange] (-130 * \d, 30 * \d) circle (\r);
\draw[fill=red] (-130 * \d, -45 * \d) circle (\r);
\draw[fill=green] (-90 * \d, -105 * \d) circle (\r);
\draw[fill=violet!70] (-30 * \d, -125 * \d) circle (\r);
\draw[fill=yellow] (45 * \d, -120 * \d) circle (\r);
\draw[fill=orange] (95 * \d, -95 * \d) circle (\r);
\draw[fill=violet!70] (130 * \d, -30 * \d) circle (\r);
\draw[fill=blue!70] (130 * \d, 45 * \d) circle (\r);
\draw[fill=red] (90 * \d, 105 * \d) circle (\r);
\draw[fill=yellow] (30 * \d, 125 * \d) circle (\r);
\draw[fill=blue!70] (-45 * \d, 120 * \d) circle (\r);
\draw[fill=green] (-95 * \d, 95 * \d) circle (\r);

\draw[fill=violet!70] (-90 * \d, 60 * \d) circle (\r);
\draw[fill=blue!70] (10 * \d, -90 * \d) circle (\r);
\draw[fill=green] (95 * \d, 70 * \d) circle (\r);
\draw[fill=green] (90 * \d, -60 * \d) circle (\r);
\draw[fill=violet!70] (-10 * \d, 90 * \d) circle (\r);
\draw[fill=blue!70] (-95 * \d, -70 * \d) circle (\r);

\draw[fill=red] (10 * \d, -25 * \d) circle (\r);
\draw[fill=orange] (45 * \d, 45 * \d) circle (\r);
\draw[fill=yellow] (-30 * \d, 40 * \d) circle (\r);
\draw[fill=orange] (-10 * \d, 25 * \d) circle (\r);
\draw[fill=yellow] (-45 * \d, -45 * \d) circle (\r);
\draw[fill=red] (30 * \d, -40 * \d) circle (\r);

\draw[fill=black] (0,0) circle (0.33 * \r);
\draw[->,>=stealth] (0,0) --  (-60 * \d, -32.5 * \d) node[anchor=south east] {$x$};
\draw[->,>=stealth] (0,0) --  (70 * \d, -25 * \d) node[anchor=south west] {$y$};
\draw[->,>=stealth] (0,0) --  (-0.1 * \d, 82.5 * \d) node[anchor=west] {$z$};

\end{scope}

\end{tikzpicture}
\caption{ Fundamental zone for rotations with octahedral symmetry. }
\label{fig:rodr}
\end{figure}
\end{center}

The topology we consider is obtained by gluing the opposite octagons and the opposite triangles with $45^\circ$ and $60^\circ$ turn respectively.
The colors in the picture indicate how the vertices map to each other.

\newpage

\bibliographystyle{plain}
\bibliography{lit}

\end{document}